\documentclass[11pt,a4paper]{article}
\pdfoutput=1
\usepackage[margin=1in]{geometry}
\usepackage{amsmath,amssymb,amsthm,booktabs,color,doi,enumitem,graphicx,url}
\usepackage{thmtools, thm-restate}
\usepackage{stmaryrd}
\usepackage[numbers,sort&compress]{natbib}
\usepackage{microtype}
\usepackage{hyperref}

\newtheorem{theorem}{Theorem}

\theoremstyle{definition}
\newtheorem{definition}{Definition}

\theoremstyle{remark}
\newtheorem{remark}{Remark}

\newcommand{\ldo}{\mathsf{LDO}}
\newcommand{\ld}{\mathsf{LD}}

\newcommand{\id}{\text{\sf id}}

\newcommand{\F}{\mathcal{F}}

\hypersetup{
    colorlinks=true,
    linkcolor=black,
    citecolor=black,
    filecolor=black,
    urlcolor=[rgb]{0,0.1,0.5},
    pdfauthor={Pierre Fraigniaud, Juho Hirvonen, Jukka Suomela},
    pdftitle={Node Labels in Local Decision},
}

\newenvironment{myabstract}
               {\list{}{\listparindent 1.5em%
                        \itemindent    \listparindent
                        \leftmargin    0pt
                        \rightmargin   0pt
                        \parsep        0pt}%
                \item\relax}
               {\endlist}

\newenvironment{mycover}
               {\list{}{\listparindent 0pt
                        \itemindent    \listparindent
                        \leftmargin    0pt
                        \rightmargin   0pt
                        \parsep        0pt}%
                \raggedright
                \item\relax}
               {\endlist}

\begin{document}

\vspace*{2ex}
\begin{mycover}
{\huge \bfseries Node Labels in Local Decision\par}
\bigskip
\bigskip

\textbf{Pierre Fraigniaud}

\nolinkurl{pierre.fraigniaud@liafa.univ-paris-diderot.fr}
\medskip

{\small Theoretical Computer Science Federation \\ CNRS and University Paris Diderot, France\par}

\bigskip

\textbf{Juho Hirvonen}

\nolinkurl{juho.hirvonen@aalto.fi}
\medskip

{\small Helsinki Institute for Information Technology HIIT, \\ Department of Computer Science, Aalto University, Finland\par}
\bigskip

\textbf{Jukka Suomela}

\nolinkurl{jukka.suomela@aalto.fi}
\medskip

{\small Helsinki Institute for Information Technology HIIT, \\ Department of Computer Science, Aalto University, Finland\par}

\end{mycover}

\bigskip
\begin{myabstract}
\noindent\textbf{Abstract.}
The role of unique node identifiers in network computing is well understood as far as \emph{symmetry breaking} is concerned. However, the unique identifiers also \emph{leak information} about the computing environment---in particular, they provide some nodes with information related to the size of the network. It was recently proved that in the context of \emph{local decision}, there are some decision problems such that (1)~they cannot be solved without unique identifiers, and (2)~unique node identifiers leak a \emph{sufficient} amount of information such that the problem becomes solvable (PODC 2013).

In this work we give study what is the \emph{minimal} amount of information that we need to leak from the environment to the nodes in order to solve local decision problems. Our key results are related to \emph{scalar oracles} $f$ that, for any given $n$, provide a multiset $f(n)$ of $n$ labels; then the adversary assigns the labels to the $n$ nodes in the network. This is a direct generalisation of the usual assumption of unique node identifiers. We give a complete characterisation of the \emph{weakest oracle} that leaks at least as much information as the unique identifiers.

Our main result is the following dichotomy: we classify scalar oracles as \emph{large} and \emph{small}, depending on their asymptotic behaviour, and show that (1)~any large oracle is at least as powerful as the unique identifiers in the context of local decision problems, while (2)~for any small oracle there are local decision problems that still benefit from unique identifiers.
\end{myabstract}

\thispagestyle{empty}
\setcounter{page}{0}
\newpage

\section{Introduction}

This work studies the role of \emph{unique node identifiers} in the context of \emph{local decision problems} in distributed systems. We generalise the concept of node identifiers by introducing \emph{scalar oracles} that choose the labels of the nodes, depending on the size of the network $n$---in essence, we let the oracle leak some information on $n$ to the nodes---and ask what is the \emph{weakest} scalar oracle that we could use instead of unique identifiers. We prove the following dichotomy: we classify each scalar oracle as \emph{small} or \emph{large}, depending on its asymptotic behaviour, and we show that the large oracles are precisely those oracles that are at least as strong as unique identifiers.

\subsection{Context and background}

The  research trends within the framework of distributed computing are most often pragmatic. Problems closely related to real world applications are tackled under computational assumptions reflecting existing systems, or systems whose future existence is plausible. Unfortunately, small variations in the model settings may lead to huge gaps in terms of computational power. Typically, some problems are unsolvable in one model but may well be efficiently solvable in a slight variant of that model. In the context of \emph{network computing}, this commonly happens depending on whether the model assumes that pairwise distinct identifiers are assigned to the nodes. While the presence of distinct identifiers is inherent to some systems (typically, those composed of artificial devices), the presence of such identifiers is questionable in others (typically, those composed of biological or chemical elements). Even if the identifiers are present, they may not necessarily be directly visible, e.g., for privacy reasons.

The absence of identifiers, or the difficulty of accessing the identifiers, limits the power of computation. Indeed, it is known that the presence of identifiers ensures two crucial properties, which are both used in the design of efficient algorithms. One such property is \textbf{symmetry breaking}. The absence of identifiers makes symmetry breaking far more difficult to achieve, or even impossible if asymmetry cannot be extracted from the inputs of the nodes, from the structure of the network, or from some source of random bits. The role of the identifiers in the framework of network computing, as far as symmetry breaking is concerned, has been investigated in depth, and is now well understood \cite{angluin80local,boldi01effective,chalopin06groupings,czygrinow08fast,diks95anonymous,emek14anonymous,emek14revocable,fraigniaud01labels,goos13local-approximation,hasemann14scheduling,hella15weak-models,lenzen08leveraging,linial92locality,mayer95local,naor95what,norris94classifying-anonymous,yamashita96computing,yamashita99leader,fich03hundreds,kranakis96symmetry,suomela13survey}.

The other crucial property of the identifiers is their ability to \textbf{leak global information} about the framework in which the computation takes place. In particular, the presence of pairwise distinct identifiers guarantees that at least one node has an identifier at least~$n$ in $n$-node networks. This apparently very weak property was proven to actually play an important role when one is interested in checking the correctness of a system configuration  in a decentralised manner. Indeed, it was shown in prior work~\cite{fraigniaud13ld-id} that the ability to check the legality of a system configuration with respect to some given Boolean predicate differs significantly according to the ability of the nodes to use their identifiers. This phenomenon is of a nature different from symmetry breaking, and is far less understood than the latter. 

More precisely, let us define a \emph{distributed language} as a set of system configurations (e.g., the set of properly coloured networks, or the set of networks each with a unique leader). Then let $\ld$ be the class of distributed languages that are \emph{locally decidable}. That is, $\ld$ is the set of distributed languages for which there exists a distributed algorithm where every node inspects its neighbourhood at constant distance in the network, and outputs \emph{yes} or \emph{no} according to the following rule:  all nodes output \emph{yes} if and only if the instance is legal. Equivalently, the instance is illegal if and only if at least one node outputs \emph{no}. Let $\ldo$ be defined as $\ld$ with the restriction the local algorithm is required to be \emph{identifier oblivious}, that is, the output of every node is the same regardless of the identifiers assigned to the nodes. By definition, $\ldo\subseteq \ld$, but~\cite{fraigniaud13ld-id} proved that this inclusion is strict: there are languages in $\ld\setminus\ldo$. This strict inclusion was obtained by constructing a distributed language that can be decided by an algorithm whose outputs depend heavily on the identifiers assigned to the nodes, and in particular on the fact that at least one node has an identifier whose value is at least~$n$.

The gap between $\ld$ and $\ldo$ has little to do with symmetry breaking. Indeed, decision tasks do not require that some nodes act differently from the others: on legal instances, all nodes must output \emph{yes}, while on illegal instances, it is permitted (but not required) that all nodes output \emph{no}. The gap between $\ld$ and $\ldo$ is entirely due to the fact that the identifiers leak information about the size $n$ of the network. Moreover, it is known that the gap between $\ld$ and $\ldo$ is strongly related to computability issues: there is an identifier-oblivious \emph{non-computable} simulation $A'$ of every local algorithm $A$ that uses identifiers to decide a distributed language~\cite{fraigniaud13ld-id}. Informally, for every language in $\ld\setminus\ldo$, the unique identifiers are precisely as helpful as providing the nodes with the capability of solving undecidable problems.

\subsection{Objective}

One objective of this paper is to measure the \emph{amount of information} provided to a distributed system via the labels given to its nodes. For this purpose, we consider  the classes $\ld$ and  $\ldo$ enhanced with \emph{oracles}, where an oracle $f$ is a function that provides every node with information about its environment.

We focus on the class of \emph{scalar} oracles, which are functions over the positive integers. Given an $n\geq 1$, a scalar oracle $f$ returns a list $f(n)=(f_1,\dots,f_n)$ of $n$ labels (bit strings) that are assigned arbitrarily to the nodes of any $n$-node network in a one-to-one manner. The class $\ld^f$ (resp., $\ldo^f$) is then defined as the class of distributed languages  decidable locally by an algorithm (resp., by an identifier-oblivious algorithm) in networks labelled with oracle~$f$.

If, for every $n\geq 1$, the $n$ values in the list $f(n)$ are pairwise distinct, then $\ld \subseteq \ldo^f$ since the nodes can use the values provided to them by the oracle as identifiers. However, as we shall demonstrate in the paper, this pairwise distinctness condition is not necessary.

Our goal is to identify the interplay between the classes $\ld$, $\ldo$, $\ld^f$, and $\ldo^f$, with respect to any scalar oracle~$f$, and to characterise  the power of identifiers in distributed systems as far as leaking information about the environment is concerned. 

\subsection{Our results}

Our first result is a characterisation of the weakest oracles providing the same power as unique node identifiers. We say that a scalar oracle $f$ is \emph{large} if, roughly, $f$ ensures that, for any set of $k$ nodes, the largest value provided by $f$ to the nodes in this set grows with $k$ (see Section~\ref{ssec:def-oracle} for the precise definition). We show the following theorem.
\begin{theorem} \label{thm:main}
    For any computable scalar oracle $f$, we have $\ldo^f = \ld^f$ if and only if $f$ is large.
\end{theorem}

Theorem~\ref{thm:main} is a consequence of the following two lemmas. The first says that small oracles (i.e.\ non-large oracles) do not capture the power of unique identifiers. Note that the following separation result holds for any small oracle, including uncomputable oracles.
\begin{restatable}{lemma}{lemweakseparation} \label{lem:weak-separation}
    For any small oracle $f$, there exists a language $L \in \ld \setminus \ldo^{f}$.
\end{restatable}
The second is a simulation result, showing that any local decision algorithm using identifiers can be simulated by an identifier-oblivious algorithm with the help of \emph{any} large oracle, as long as the oracle itself is computable. Essentially large oracles capture the power of unique identifiers.
\begin{restatable}{lemma}{lemstrongldofldf} \label{lem:strong-ldof-ldf}
    For any large computable oracle $f$, we have $\ld \subseteq \ldo^f = \ld^f$.
\end{restatable}

Theorem~\ref{thm:main} holds despite the fact that small oracles can still produce some large values, and that there exist small oracles guaranteeing that, in any $n$-node network, at least one node has a value at least $n$. Such a small oracle would be sufficient to decide the language $L \in \ld \setminus \ldo$ presented in~\cite{fraigniaud13ld-id}. However, it is not sufficient to decide all languages in $\ld$.

Our second result is a complete description of the hierarchy of the four classes $\ld$, $\ldo$, $\ld^f$, and $\ldo^f$ of local decision, using identifiers or not,  with or without oracles. The pictures for small and large oracles are radically different. 
\begin{itemize}
\item  For any large oracle $f$, the hierarchy yields a \emph{total order}:
\[
\ldo \subsetneq \ld \subseteq \ldo^f=\ld^f.
\]
The strict inclusion $\ldo \subsetneq \ld$ follows from~\cite{fraigniaud13ld-id}. The second inclusion $\ld \subseteq \ldo^f$ may or may not be strict depending on oracle $f$. 

\item For any small oracle $f$, the hierarchy yields a \emph{partial order}. We have $\ldo^f\subsetneq \ld^f$ as a consequence of Lemma~\ref{lem:weak-separation}. However, $\ld$ and $\ldo^f$ are incomparable, in the sense that there is a language $L \in \ld \setminus \ldo^f$ for any small oracle $f$, and there is a language $L \in \ldo^f \setminus \ld$ for some small oracles $f$. Hence, the relationships of the four classes can be represented as the following diagram: 
\[
\begin{array}{c@{}c@{}c@{}c@{}c}
  && \ld^f && \\
  & \nearrow &&  \nwarrow & \\[3pt]
 \ldo^f  &&&& \ld \\
  & \nwarrow && \nearrow & \\
  && \ldo &&
\end{array}
\]
All inclusions (represented by arrows) can be strict. 
\end{itemize}

\subsection{Additional related work}

In the context of network computing, oracles and advice commonly appear in the form of \emph{labelling schemes} \cite{gavoille03compact,fraigniaud07distributed}. A typical example is a \emph{distance labelling scheme}, which is a labelling of the nodes so that the distance between any pair of nodes can be computed or approximated based on the labels. Other examples are \emph{routing schemes} that label the nodes with information that helps in finding a short path between any given source and destination. For graph problems, one could of course encode the entire solution in the advice string---hence the key question is whether a very small amount of advice helps with solving a given problem.

In prior work, it is commonly assumed that the oracle can give a specific piece of advice for each individual node. The advice is localised, and entirely controlled by the oracle. Moreover, the oracle can see the entire problem instance and it can tailor the advice for any given task.

In the present work, we study a much weaker setting: the oracle is only given $n$, and it cannot choose which label goes to which node. This is a generalisation of, among others, typical models of \emph{networks with unique identifiers}: one commonly assumes that the unique identifiers are a permutation of $\{1,2,\dotsc,n\}$ \cite{linial92locality}, which in our case is exactly captured by the large scalar oracle
\[
    f(n) = (1,2,\dotsc,n),
\]
or that the unique identifiers are a subset of $\{1,2,\dotsc,n^c\}$ for some constant $c$ \cite{peleg00distributed}, which in our case is captured by a subfamily of large scalar oracles. Our model is also a generalisation of \emph{anonymous networks with a unique leader} \cite{fraigniaud01labels}---the assumption that there is a unique leader is captured by the small scalar oracle
\[
    f(n) = (0,0,\dotsc,0,1).
\]

\section{Model and definitions}

In this work, we augment the usual definitions of \emph{locally checkable labellings} \cite{naor95what} and \emph{local distributed decision} \cite{fraigniaud11ld,fraigniaud12impact,fraigniaud13ld-id} with scalar oracles.

\subsection{Computational model}

We deal with the standard \textsf{LOCAL} model~\cite{peleg00distributed} for distributed graph algorithms. In this model, the network is a simple connected graph $G = (V,E)$. Each node $v\in V$ has an \emph{identifier} $\id(v)\in\mathbb{N}$, and all identifiers of the nodes in the network are pairwise distinct. Computation proceeds in synchronous rounds. During a round, each node communicates with its neighbours in the graph, and performs some local computation. There are no limits to the amount of communication done in a single round. Hence, in $r$ communication rounds, each node can learn the complete topology of its radius-$r$ neighbourhood, including the inputs and the identifiers of the nodes in this neighbourhood. In a distributed algorithm, all nodes start at the same time, and each node must halt after some number of rounds, and produce its individual output. The collection of individual outputs then forms the global output of the computation. The running time of the algorithm is the number of communication rounds until all nodes have halted. 

We consider \emph{local} algorithms, i.e., constant-time algorithms~\cite{suomela13survey}. That is, we focus on algorithms with a running time that does not depend on  the size $n$ of the graph. Any such algorithm, with running time $r$, can be seen as a function from the set of all possible radius-$r$ neighbourhoods to the set of all possible outputs. An \emph{identifier-oblivious} algorithm is an algorithm whose outputs are independent of the identifiers assigned to the nodes. Note that, from the perspective of an identifier-oblivious algorithm, the set of all possible radius-$r$ degree-$d$ neighbourhoods is finite. This is not the case for every algorithm since there are infinitely many identifier assignments to the nodes in a radius-$r$  degree-$d$ neighbourhood. 

Although the \textsf{LOCAL} model does not put any restriction on the amount of individual computation performed at each node, we only consider algorithms that are \emph{computable}. 

\subsection{Local decision tasks} 

We are interested in the power of constant-time algorithms for \emph{local decision}. A \emph{labelled graph} is a pair $(G, x)$, where $G$ is a simple connected graph, and $x:V(G)\to\{0,1\}^*$ is a function assigning a label to each node of $G$. A \emph{distributed language} $L$ is a set of labelled graphs.  Examples of distributed languages include:
\begin{itemize}[noitemsep]
    \item 2-colouring, the language where $G$ is a bipartite graph and $x(v) \in \{0,1\}$ for all $v \in V(G)$ such that $x(v) \neq x(u)$ whenever $\{u,v\} \in E(G)$;
    \item parity, the language of graphs with an even number of nodes;
    \item planarity, the language that consists of all planar graphs.
\end{itemize}

We say that algorithm $A$ decides $L$ if and only if the output of $A$ at every node is either  \emph{yes} or  \emph{no}, and, for every instance $(G,x)$, $A$ satisfies:  
\[(G,x) \in L  \iff \; \mbox{all nodes output \emph{yes}.}
\]
Hence, for an instance $(G,x)\notin L$, the algorithm $A$ must ensure that at least one node outputs \emph{no}. We consider two main distributed complexity classes: 
\begin{itemize}
\item $\ld$ (for \emph{local decision}) is the set of languages decidable by constant-time algorithms in the \textsf{LOCAL} model. 
\item $\ldo$  (for \emph{local decision oblivious}) is the set of languages decidable by constant-time identifier-oblivious algorithms in the \textsf{LOCAL} model.
\end{itemize}
By definition, $\ldo\subseteq \ld$, and it is known~\cite{fraigniaud13ld-id} that this inclusion is strict: there are languages $L\in\ld\setminus\ldo$. The fact that we consider only computable algorithms is crucial here---without this restriction we would have $\ldo = \ld$~\cite{fraigniaud13ld-id}.

\subsection{Distributed oracles}\label{ssec:def-oracle}

We study the relationship of classes $\ld$ and $\ldo$ with respect to \emph{scalar oracles}. Such an oracle $f$ is a function that assigns a list of $n$ values to every positive integer $n$, i.e., 
\[
f(n) = (f_1, f_2, \dots, f_n)
\] 
with $f_i\in\{0,1\}^*$. In essence, oracle $f$ can provide some information related to $n$ to the nodes. In an $n$-node graph, each of the $n$ nodes will receive a value $f_i\in f(n)$, $i\in[n]$. These values are arbitrarily assigned to the nodes in a one-to-one manner. Two different nodes will thus receive $f_i$ and $f_j$ with $i\neq j$. Note that $f_i$ may or may not be different from $f_j$ for $i\neq j$; this is up to the choice of the oracle. The way the values provided by the oracles are assigned to the nodes is under the control of an adversary.  One example of an oracle is $f(n)=(1,2,\dots,n)$, which provides the nodes with identifiers. Another example is $f(n)=(0,0,\dots,0)$, which provides no information to the nodes. 

W.l.o.g., let us assume that $f_i \le f_{i+1}$ for every $i$. We use the shorthand $f^{(n)}_k$ for the $k$th label provided by $f$ on input $n$, that is, $f(n) = (f^{(n)}_1, f^{(n)}_2, \dotsc, f^{(n)}_n)$. For a fixed oracle $f$, we consider two main distributed complexity classes: 
\begin{itemize}[noitemsep]
\item $\ld^f$ is the set of languages decidable by constant-time algorithms in networks that are labelled with oracle~$f$.
\item $\ldo^f$ is the set of languages decidable by constant-time indentifier-oblivious algorithms in networks that are labelled with oracle~$f$.
\end{itemize}
We will separate oracles in two classes, which play a crucial role in the way the four classes $\ldo$, $\ld$, $\ldo^f$, and $\ld^f$ interact.

\begin{definition}\em
An oracle $f$ is said to be \emph{large} if $$\forall c > 0,\, \exists k \geq 1,\, \forall n \ge k, \, f_k^{(n)} \ge c.$$ An oracle is \emph{small} if it is not large. 
\end{definition}

Hence, a large oracle $f$ satisfies that,  for any value $c > 0$, there exists a large enough $k$, such that, in every graph $G$ of size at least $k$, for every set of nodes $S\subseteq V(G)$ of size $|S|\geq k$, oracle $f$ is providing at least one node of $S$ with a value at least as large as $c$. In short: every large set of nodes must include at least one node that receives a large value. 

Conversely, a small oracle $f$ satisfies that there exists a value  $c>0$ such that, for every $k$, we can find $n\geq k$ such that, in every $n$-node graph $G$, and for every set of nodes $S\subseteq V(G)$ of size $|S|\geq k$, there is an assignment of the values provided by $f$ such that every node in $S$ receives a value smaller that~$c$.  In short: there are arbitrarily  large sets of nodes which all receive a small value. 

For example, oracles $f(n) = (1,2,\dotsc,n)$ and $f(n) = (n,n,\dotsc,n)$ are large, while oracles $f(n) = (0,0,\dotsc,0,1)$ and $f(n) = (0,0,\dotsc,0,2^n)$ are small. We emphasise that small oracles can output very large values.

\section{Proof of the main theorem}

In this section we give the proof of our main result that characterises the power of weak and large oracles with respect to identifier-oblivious local decision.

\subsection{Small oracles do not capture the power of unique identifiers} \label{ssec:lem-1}

Fraigniaud et al.~\cite{fraigniaud13ld-id} showed that there exists a language $L \in \ld \setminus \ldo$. We use a very similar Turing machine construction as in the proof of their Theorem~1. However, we must take into account the additional concern of the values that the oracle assigns to the nodes. We handle this by forcing any small oracle to always give many copies of the same constant label $c$ so that the adversary can cover the interesting parts of the construction with this unhelpful label $c$. We can then use uncomputability arguments to show that if a certain language were in $\ldo^f$, then we could get a sequential algorithm for uncomputable problems. See Figure~\ref{fig:lem-1} for illustrations.

\begin{figure}[t]
    \centering
    \includegraphics[page=2]{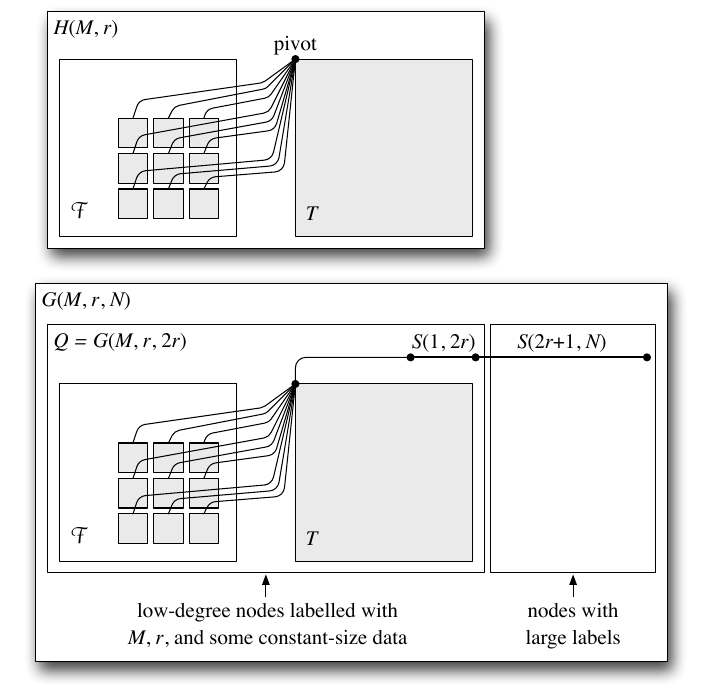}
    \caption{The construction of Section~\ref{ssec:lem-1}.}\label{fig:lem-1}
\end{figure}

\lemweakseparation*

\begin{proof}
We assume that for each halting Turing machine $M$ and each locality parameter $r \in \mathbb{N}$, there exists a labelled graph $H(M,r)$ with the following properties:
\begin{enumerate}[label=(P\arabic*),leftmargin=9mm]
    \item There is an identifier-oblivious local checker that verifies that a given labelled graph is a equal to $H(M,r)$ for some $M$ and $r$. \label{con:p1}
    \item The number of nodes in the graph $H(M,r)$ is at least as large as the number of steps $M$ takes on an empty tape. \label{con:p2}
    \item Given $H(M,r)$, an identifier-oblivious local checker $A$ with a running time of $r$ cannot decide if $M$ outputs $0$ or $1$. \label{con:p3}
    \item Each label of $H(M,r)$ is a triple $x(v) = (M,r,x'(v))$. The maximum degree of $H$ and the maximum size of $x'(v)$ are constants that only depend on $r$. \label{con:p4}
    \item Graph $H(M,r)$ can be padded with additional nodes without violating properties \ref{con:p1}--\ref{con:p4}.
\end{enumerate}

The construction of Fraigniaud et al.~\cite{fraigniaud13ld-id} satisfies these properties. They show how to construct a labelled graph $H(M,r)$ that encodes the execution table of a given Turing machine $M$ such that a local checker with running time $r$ cannot decide if $M$ halts with $0$ or $1$. The original construction $(H, x) = H(M,r)$ consists of three main parts. 
\begin{enumerate}[label=(\roman*),leftmargin=9mm]
    \item \emph{The execution table $T$ of the Turing machine $M$}. Let $s$ be the number of steps $M$ takes on an empty tape. Then table $T$ is an $(s+1) \times (s+1)$ grid, where node $(i,j)$ holds the contents of the tape at position $j$ after computation step $i$, and its own coordinates $(i,j)$ modulo~$3$. Node $(i,j)$ also knows if the head is at position $j$ after step $i$, and if so, what is the state of $M$ after step $i$. Node $(0,0)$ representing the first position of the empty tape is called the \emph{pivot}. The execution table exists essentially to guarantee \ref{con:p2}. \label{part-1}
    \item \label{part-2} \emph{The fragment collection $\F$.} This is a collection of subgrids labelled with all syntactically possible ways that are consistent with being in some execution table of $M$. The dimensions of the fragments are linear in $r$ and independent of $M$. In each fragment, every $2\times2$ subgrid is consistent with a state transition of $M$. It is crucial to observe that there is a finite number of such fragments. Each fragment is connected to the pivot in a way that supports the local verification of the structure. The fragment collection is added to ensure \ref{con:p3}. Informally, if we only had $T$, then some node $(i,s)$ at the last row of the grid would be able to see the stopping state of $M$; however, $\F$ will contain some fragments in which $M$ halts with output $0$ and some fragments in which $M$ halts with output $1$, and the nodes at the last row of $T$ are locally indistinguishable from the nodes in such fragments.
    \item \emph{Pyramid structure}. This is added to the execution table and to the fragments to ensure \ref{con:p1}. Without any additional structure, a grid with coordinates modulo $3$ is locally indistinguishable from, e.g., a grid that is wrapped into a torus. The pyramid structure guarantees that at least one node is able to detect invalid instances.
\end{enumerate}

Finally, since all labellings can be made constant-size, we can ensure \ref{con:p4}. In particular, for any $(M,r)$, there are constantly many syntactically possible $r$-neighbourhoods of $H(M,r)$. This is a crucial property as it guarantees that there is a sequential algorithm that on all inputs $(M,r)$ halts and, if $M$ halts, outputs all possible labelled $r$-neighbourhoods of $H(M,r)$.

Let $S(a,b)$ be the labelled path $(s_a, s_{a+1}, \dots, s_b)$ in which node $s_i$ is labelled with value $i$. We augment the construction $H(M,r)$ as follows: labelled graph $G(M,r,N)$ consists of $H(M,r)$, plus $S(1,N)$, plus an edge between the pivot of $H(M,r)$ and the first node $s_1$ of the path $S(1,N)$; we call $S(1,N)$ the \emph{tail} of the construction. The structure of $G(M,r,N)$ is still locally checkable in $\ldo$: any tail must eventually connect to the pivot, and the pivot can detect if there are multiple tails. The key property of the construction is that the nodes in the tail $S(1,N)$ with large labels are far from the nodes of $G(M,r)$ that are aware of $M$.

We will separate $\ld$ and $\ldo^f$ using the following language:
\[
    L = \{ G(M, r, N) : r \ge 1,\ N \ge 1, \text{ and Turing machine } M \text{ outputs } 0 \}.
\]
We have $L \in \ld$ as there will be a node $v$ with $\id(v) \ge s$ which can simulate $M$ for $s$ steps and output \emph{no} if $M$ does not output $0$. Next we will argue that $L$ cannot be in $\ldo^f$ for any small~$f$.

Let $f$ be a small oracle. For any $M$ and $r$, we can choose a sufficiently large $N$ as follows. By definition, there exists a $c$ such that for all $k$ oracle $f$ outputs some label $i \in [c]$ at least $\lceil k/c \rceil$ times on some $n \ge k$. Moreover, we can find an infinite sequence of values $k_0, k_1, \dots$ such that the most common value is some fixed $i_0$. We select w.l.o.g.\ the smallest $k_j$ and a suitable $n$ such that $f(n)$ contains at least $k_j/c \ge |H(M,r)| + 2r$ labels equal to $i_0$. Let $N = n-|H(M,r)|$, and consider $G(M,r,N)$. Now the adversary can construct the following \emph{worst-case labelling}: every node of $G(M,r,2r) \subseteq G(M,r,N)$ receives the constant input $i_0 \in [c]$; all other labels as assigned to the nodes in $S(2r+1,N) \subseteq G(M,r,N)$.

It is known that separating the following languages is undecidable (see e.g.~\cite[p.~65]{papadimitriou94computational}):
\begin{equation}\label{eq:Li}
L_i = \{ M : \text{Turing machine } M \text{ outputs } i \} : i \in \{0,1\}.
\end{equation}
For the sake of contradiction, we assume that there is an $\ldo^f$-algorithm $A$ that decides $L$. We will use algorithm $A$ and constant $i_0$ defined above to construct a sequential algorithm $B$ that separates $L_0$ and $L_1$.

Let $r$ be the running time of $A$, and consider the execution of $A$ on an instance $G(M,r,N)$ for some $M$ and $N$. It follows that each node in $S(r+1,N) \subseteq G(M,r,N)$ must always output \emph{yes}. To see this, note that the claim is trivial if $M$ halts with $0$. Otherwise we can always construct another instance $G(M_0,r,N)$ such that $M_0$ halts with $0$ and both $G(M,r,N)$ and $G(M_0,r,N)$ have the same number of nodes. Hence the oracle and the adversary can assign the same labels to $S(r+1,N)$ in both $G(M,r,N)$ and $G(M_0,r,N)$. If any of these nodes would answer \emph{no} in $G(M,r,N)$, then $A$ would also incorrectly reject the \emph{yes}-instance $G(M_0,r,N) \in L$.

Now given a Turing machine $M$, algorithm $B$ proceeds as follows. Consider the subgraph $Q = G(M,r,2r) \subseteq G(M,r,N)$, and assume the worst-case labelling of $G(M,r,N)$ in which all nodes of $Q$ have the constant label $i_0$. Algorithm $B$ cannot construct $Q$; indeed, $M$ might not halt, in which case $G(M,r,N)$ would not even exist. However, $B$ can do the following: it can assume that $M$ halts, and then generate a collection $\mathcal{Q}$ that would contain all possible radius-$r$ neighbourhoods of the nodes in $G(M,r,r)$. Collection $\mathcal{Q}$ is finite, its size only depends on $r$ and $M$, and the key observation is that $\mathcal{Q}$ is computable (in essence, $B$ enumerates all syntactically possible fixed-size fragments of partial execution tables of $M$).

Then $B$ will simulate $A$ in each neighbourhood of $\mathcal{Q}$. If $M$ halts with $1$, then $G(M,r,N) \notin L$, and therefore one of the nodes in $G(M,r,r)$ has to output \emph{no}; in this case $B$ outputs $1$. If $M$ halts with $0$, then $G(M,r,N) \in L$, and therefore one of the nodes in $G(M,r,r)$ has to output \emph{yes}; in this case $B$ outputs $0$. The key observation is that $B$ will always halt with some (meaningless) output even if we are given an input $M \notin L_0 \cup L_1$; hence $B$ is a computable function that separates $L_0$ and $L_1$. As such a $B$ cannot exist, $A$ cannot exist either. \qedhere
\end{proof}

\subsection{Large oracles capture the power of unique identifiers} \label{ssec:thm-equality}

In this section we will show that a \emph{computable} large oracle $f$ is sufficient to have $\ld \subseteq \ldo^f = \ld^f$. This result holds even if $f$ only has access to an upper bound $N \ge n$, and the adversary gets to pick an $n$-subset of labels from $f(N)$. Note that the oracle has to be computable in order for us to invert it locally.

\lemstrongldofldf*

\begin{proof}
We begin by showing how to recover an oracle $\hat{f}$ with $\hat{f}_k^{(N)} \ge k$, for all $k$ and $N \ge k$, from a large oracle $f$. We want to guarantee that each node $v$ receives a label $\ell \ge i$ if in the initial labelling it had the $i$th smallest label.

By definition, it holds for large oracles that for each natural number $\ell$ there is a largest index $i$ such that $f^{(N)}_i \le \ell$; we denote the index by $g(\ell)$. By assumption, a node with label $\ell$ can locally compute the value $g(\ell)$. We now claim that 
\[
    \hat{f}\colon N \mapsto \bigl\{ g(f_1), g(f_2), \dots, g(f_N) \bigr\}
\]
has the property $\hat{f}_k^{(N)} \ge k$. To see this, assume that we have $f^{(N)}_k = \ell$ for an arbitrary $k$. Seeing label $\ell$, node $v$ knows that, in the worst case, its own label is the $g(\ell)$th smallest. Thus for every $k$, the node with the $k$th smallest label will compute a new label at least $k$.

Now given $\hat{f}$, we can simulate any $r$-round $\ld$-algorithm $A$ as follows.
\begin{enumerate}
    \item Each node $v$ with label $\ell_v$ locally computes the new label $g(\ell_v)$.
    \item Each node gathers all labels $g(\ell_u)$ in its $r$-neighbourhood. Denote by $g_v^*$ the maximum value in the neighbourhood of $v$.
    \item Each node $v$ simulates $A$ on every unique identifier assignment to its local $r$-neighbourhood from $\{1,2,\dotsc,g_v^*\}$. If for some assignment $A$ outputs \emph{no}, then $v$ outputs \emph{no}, and otherwise it outputs \emph{yes}. 
\end{enumerate}
Because of how the decision problem is defined, it is always safe to output \emph{no} when some simulation of $A$ outputs \emph{no}. It remains to be argued that it is safe to say \emph{yes}, if all simulations say \emph{yes}. This requires that \emph{some} subset of simulations of $A$, one for each node, looks as if there had been a consistent setting of unique identifiers on the graph. Now let $\id$ be one identifier assignment with $\id(v) = i$ for the $v$ with $i$th smallest label, for all $i$ (breaking ties arbitrarily). Since by construction $g(\ell_v) \ge \id(v)$ for all $v$, there will be  a simulation of $A$ for every node $v$ with local identifier assignment $\id_v$ such that for all $u$ in the radius-$r$ neighbourhood of $v$ we have $\id_v(u) = \id(u)$.

So far we have seen how to simulate any $\ld$-algorithm $A$ with $\ldo^f$-algorithms. We can apply the same reasoning to simulate any $\ld^f$-algorithm $A$ with $\ldo^f$-algorithms; the only difference is that each node in the simulation has now access to the original oracle labels as well. \qedhere
\end{proof}

\section{Full characterisation of \texorpdfstring{$\ld^f$, $\ldo^f$, $\ld$, and $\ldo$}{LDf, LDOf, LD, and LDO}}

Our goal in this section is to complete the characterisation of the power of scalar oracles with respect to the classes $\ld$ and $\ldo$. We aim at giving a robust characterisation that holds also for minor variations in the definition of a scalar oracle. In particular, all of the key results can be adapted to weaker oracles that only receive an upper bound $N \ge n$ on the size of the graph.

\subsection{Large oracles can be stronger than identifiers}

Let us first consider large oracles. By prior work~\cite{fraigniaud13ld-id} and Lemma~\ref{lem:strong-ldof-ldf} we already know that for any computable large oracle $f$ we have a linear order
\[
    \ldo \subsetneq \ld \subseteq \ldo^f=\ld^f.
\]
Trivially, there is a large computable oracle $f(n) = (1,2,\dotsc,n)$ such that
\[
    \ldo \subsetneq \ld = \ldo^f=\ld^f.
\]
We will now show that there is also a large computable oracle $f$ such that
\[
    \ldo \subsetneq \ld \subsetneq \ldo^f=\ld^f.
\]

For a simple proof, we could consider the large oracle $f(n) = (n,n,\dotsc,n)$. Now the parity language $L$ that consists of graphs with an even number of nodes is clearly in $\ldo^f$ but not in $\ld$. However, this separation is not robust with respect to minor changes in the model of scalar oracles. In particular, if the oracle only knows an upper bound on $n$, we cannot use the parity language to separate $\ldo^f$ from $\ld$.

In what follows, we will show that the \emph{upper bound oracle} $f$ that labels all nodes with some upper bound on $N \ge n$ can be used to separate $\ldo^f$ from $\ld$. We resort again to computability arguments. The construction that we use in the proof has some elements that we do not need here (in particular, the fragment collection and the parameter $r$), but it enables us to directly reuse the same tools that we already introduced in the proof of Lemma~\ref{lem:weak-separation}.

\begin{restatable}{theorem}{thmldofldsep} \label{thm:ldof-ld-sep}
    For the upper bound oracle $f$ there exists a language $L$ such that $L \in \ldo^f \setminus \ld$.
\end{restatable}

\begin{proof}
  The following language is not in $\ld$ but is in $\ldo^f$. Recall the Turing machine construction $H(M,r)$ from Lemma~\ref{lem:weak-separation}. We augment it so that the pivot receives an extra label $\ell \in \{0,1\}$; let us denote such a construction by $J(M,r,\ell)$. Let
    \[
        L = \{ J(M,r,\ell) : r \ge 1,\ \ell \in \{0,1\}, \text{ and Turing machine } M \text{ halts outputs } \ell \}.
    \]
    First, observe that $L$ is in $\ldo^f$. Checking the structure of $H(M,r)$ and hence $J(M,r,\ell)$ is known to be in $\ldo$. Since the execution table of $M$ is contained in $J(M,r,\ell)$, it must halt within $n \le N$ steps. Finally, since the pivot is guaranteed to receive an $N\ge n$ as its oracle label, it can simulate $M$ for at most $N$ steps and determine whether $M$ halts with output~$\ell$.

    Next, we show that $L \notin \ld$. Suppose otherwise. Fix a local verifier $A$ that decides $L$, and let $r$ be the running time of~$A$. Consider a node $v$ that is within distance more than $r$ from the pivot. For such a node, algorithm $A$ must always output \emph{yes}---otherwise we could change the input label $\ell$ so that an answer \emph{no} is incorrect. Thus one of the nodes within distance $r$ from the pivot must be able to tell whether $J(M,r,\ell)$ is a \emph{no} instance.

    Using $A$, we can now design a sequential algorithm $B$ that solves the undecidable problem of separating the languages $L_i$ from~\eqref{eq:Li}. Given a Turing machine $M$, algorithm $B$:
    \begin{itemize}[noitemsep]
        \item constructs $J(M,r,1)$ up to distance $2r$ from the pivot,
        \item assigns the unique identifiers arbitrarily in this neighbourhood,
        \item simulates $A$ for each node within distance $r$ from the pivot.
    \end{itemize}
    Note that $B$ can essentially simulate $M$ for $2r$ steps to construct $J(M,r,1)$ up to distance $2r$ from the pivot; the construction is correct if $M$ halts, and it terminates even if $M$ does not halt. Now $J(M,r,1)$ is a \emph{no}-instance if and only if $M$ halts with output $0$. In this case one of the nodes within distance $r$ from the pivot has to output \emph{no}; otherwise all of them have to output \emph{yes}. In the former case $B$ outputs $0$, otherwise it outputs $1$. Clearly $B$ outputs $\ell$ for each $M \in L_\ell$. However, such an algorithm $B$ cannot exist. Therefore $A$ cannot exist, either, and we have $L \notin \ld$. \qedhere
\end{proof}

\begin{remark}
    The construction that we use above has some additional elements that were not necessary (in particular, the fragment collection and parameter $r$). However, this construction made it possible to directly reuse the same tools that we already introduced in the proof of Lemma~\ref{lem:weak-separation}.
\end{remark}

\subsection{Small oracles and identifiers are incomparable}

\begin{figure}[t]
    \centering
    \includegraphics[page=1]{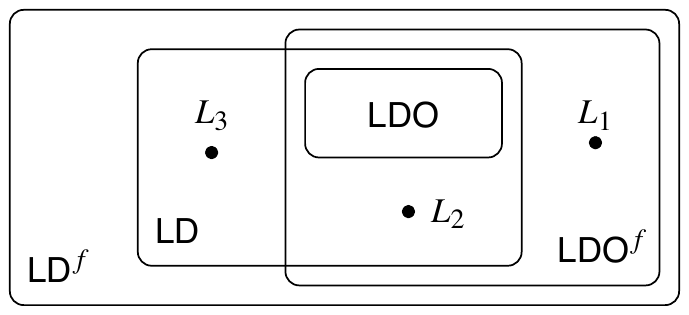}
    \caption{There is a small oracle $f$ such that each of the languages $L_i$ exists.}\label{fig:small}
\end{figure}

In the case of small oracles, we already know that $\ldo^f \subsetneq \ld^f$ for any small oracle $f$ by Lemma~\ref{lem:weak-separation}. Next we characterise the relationship of $\ldo^f$ and $\ld$. In essence, we show that these classes are incomparable.

\begin{theorem}\label{thm:small-ldof-ld}
    There is a single small oracle $f$ so that each of the languages $L_1$, $L_2$, and $L_3$ shown in Figure~\ref{fig:small} exist.
\end{theorem}

\begin{proof}
    Let $f$ be the small oracle
    \[
        f(n) = (0,0,\dotsc,0,b_n),
    \]
    where $b_n$ is an $n$-bit string such that the $i$th bit tells whether the $i$th Turing machine halts. We construct the languages as follows:
    \begin{itemize}
        \item[$L_1$:]
            Let $P(n)$ denote the labelled path of length $n$ such that each node has two input labels: $n$ and the distance to a specified leaf node $v_0$. The correct structure of $P(n)$ is in $\ldo$. Now let
            \[
                L_1 = \{ P(M) : \text{Turing machine } M \text{ halts} \}. 
            \]
            The node that receives the $n$-bit oracle label can use it to decide whether the $n$th Turing machine halts, and therefore $L_1 \in \ldo^f$. Conversely, we have $L_1 \notin \ld$; otherwise we would have a sequential algorithm that solves the halting problem for each Turing machine $M$ by constructing the path $P(M)$ with some fixed identifier assignment and simulating the local verifier.
        \item[$L_2$:]
            We can use the same language
            \[
                L_2 = \{ H(M,r) : r \ge 1 \text{ and Turing machine } M \text{ outputs 0}\}
            \]
            that we used in the proof of Lemma~\ref{lem:weak-separation}. It is known that $L_2 \in \ld$ and $L_2 \notin \ldo$ \cite{fraigniaud13ld-id}. Since checking the structure of $H(M,r)$ is in $\ldo$, it suffices to note that the node that receives the bit vector $b_n$ of length $n$ can use the \emph{length} of the vector as an upper bound in simulating $M$. Thus $L_2 \in \ldo^f$.
        \item[$L_3$:]
            Apply Lemma~\ref{lem:weak-separation}. \qedhere
    \end{itemize}
\end{proof}

We conclude by noting that Theorem~\ref{thm:small-ldof-ld} is also robust to minor variations in the definitions. In particular, the oracle does not need to know the exact value of $n$; it is sufficient that at least one node receives the bit string $b_N$, where $N \ge n$ is some upper bound on $n$.

\section*{Acknowledgements}

Thanks to Laurent Feuilloley for discussions.

\bibliographystyle{plainnat}
\bibliography{articles}

\end{document}